\documentclass[a4paper]{article}
\usepackage{amscd}
\usepackage{fancyhdr,latexsym,amsmath,amsfonts,amssymb,amsbsy,amsthm,url}%
\usepackage{natbib}
\usepackage[hscale=0.7,centering]{geometry}
\usepackage{graphicx,epsfig, xy}
\usepackage{hyperref}
\usepackage{lscape,fancybox}
\usepackage{color}
\numberwithin{equation}{section}

\numberwithin{equation}{section}

\newcommand{\BE}{{\mathbb{E}}}

\newtheorem{theorem}{Theorem}

\title{\textbf{Measuring spatial association and testing spatial independence based on short time course data
}}

\author{Divya Kappara}
\date{}

\begin{document}
	\def\shortauthors{ABC}
	\author{
		\hspace{0.01\textwidth}
		\parbox[t]{0.25\textwidth}{{Divya Kappara}
			\\ {\small University of Hyderabad\\
					and \\
					IIT-Bombay\\
					kapparadivya@gmail.com\\
				\\ }}
		\hspace{0.01\textwidth}
		\parbox[t]{0.3\textwidth}{{Arup Bose}
			\thanks{Research  supported by J.C.~Bose National Fellowship, Dept.~of Science and Technology, Govt.~of India.}
			\\ {\small Stat-Math Unit\\
			Indian Statistical Institute\\
				bosearu@gmail.com\\
				\\ }}
		\hspace{0.01\textwidth}
		\parbox[t]{0.40\textwidth}{{Madhuchhanda Bhattacharjee}
				\\ {\small University of Hyderabad\\
					and \\
					University of Manchester\\
				chhanda.bhatta@gmail.com\\
					\\ }}
		}

	\author{
		\hspace{0.01\textwidth}
		\parbox[t]{0.27\textwidth}{{Divya Kappara}
			\\ {\small University of Hyderabad\\
					and \\
					IIT-Bombay\\
					kapparadivya@gmail.com\\
				\\ }}
		\hspace{0.01\textwidth}
		\parbox[t]{0.25\textwidth}{{Arup Bose}
			\thanks{Research  supported by J.C.~Bose National Fellowship, Dept.~of Science and Technology, Govt.~of India.}
			\\ {\small Stat-Math Unit\\
			Indian Statistical Institute\\
				bosearu@gmail.com\\
				\\ }}
		\hspace{0.01\textwidth}
		\parbox[t]{0.40\textwidth}{{Madhuchhanda Bhattacharjee}
				\\ {\small University of Hyderabad\\
					and \\
					University of Manchester\\
				chhanda.bhatta@gmail.com\\
					\\ }}
		}


	\maketitle

	\begin{abstract} 
	\noindent Spatial association measures for univariate static spatial data are widely used. When the data is in the form of a collection of spatial vectors with the same temporal domain of interest, we construct a measure of similarity between the regions' series, using Bergsma's correlation coefficient $\rho$. Due to the special properties of $\rho$, unlike other spatial association measures which test for \textit{spatial randomness}, our statistic can account for \textit{spatial pairwise independence}. We have derived the asymptotic behaviour of our statistic under null (independence of the regions) and alternate cases (the regions are dependent). We explore the alternate scenario of spatial dependence further, using simulations for the SAR and SMA dependence models. Finally, we provide  application to modelling and testing for the presence of spatial association in COVID-19 incidence data, by using our statistic on the residuals obtained after model fitting.
		
	\end{abstract}
	\noindent \textbf{Keywords:} Bergsma's correlation, spatial association measure, $U$-statistic, spatial autoregressive model, spatial moving average model.\vskip5pt
	
	\noindent \textbf{AMS Subject classification}: 
	Primary 62H20; 
	Secondary  62F12, 
             92D30, 
             62H11, 
	           62P10, 
			       62M30.

\section{Introduction}\label{sec:intro}
	\textit{Spatial association} refers to the study of relatedness over space. Long before the term \textit{spatial autocorrelation} was formally introduced, 	\cite{ravenstein1885laws}, \cite{zipf1946p}, and \cite{ullman1956role} have addressed the idea of distance decay effect over space and human spatial interaction. 
The first law of geography by \cite{tobler1970computer}, supporting the fact that georeferenced observations are not generally independent of one another states that ``\textit{Everything is related to everything else, but near things are more related than distant things}''. 
	
	The primary reason behind the development of spatial association measures is the realization that it is unrealistic to assume stationarity over space.	Many prominent researchers have developed such measures. 	One of the most widely used measures is Moran's $I$ (\cite{moran1950notes}), which is based on a global covariance representation. \cite{cliff1972testing} generalized Moran's work by demonstrating how one can test residuals of regression analysis for \textit{spatial randomness}. They have worked out moments of Moran's $I$ and its distributional properties under varying sampling assumptions. Geary's $c$ (\cite{geary1954contiguity}) is another measure that has a global differences representation. 
	
	\cite{anselin1995local} introduced local measures of spatial association (LISA) to detect variations across space in the presence of spatial heterogeneity. \cite{getis2007reflections, getis2008history} provide excellent reviews on the study of spatial autocorrelation. 	\cite{getis2010analysis} introduced a family of statistics $G$ to detect the local clusters of dependence. Using estimates of spatial autocorrelation coefficients in regression models, is a widely used technique in spatial econometrics, see \cite{anselin1988spatial}. 
	
Suppose data is collected on a single variable at different time points across $R$ spatial units, yielding $R$ concurrent spatial time series.
Many researchers have attempted to provide measures of spatio-temporal dependence. \cite{martin1975identification}, \cite{shao2013analysis}  and \cite{gao2019measuring} etc., have extended univariate Moran's $I$ to the Spatio-temporal Moran's $I$. \cite{dube2013spatio} have formulated the idea of constructing a spatio-temporal weights matrix by joining two independently constructed spatial and temporal proximity matrices.
	
	In Section \ref{sec:spatialassoc} we discuss in detail spatial association measures and their components. Any spatial measure of association requires a measure of similarity between the recorded observations of a random variable. Euclidean distance, Frechet distance, Pearson's correlation, are some of the similarity measures that are often used. \cite{bergsma2006new} introduced a correlation coefficient $\rho$, as a measure of independence such that, $\rho=0$  if and only if $X$ and $Y$ are independent. 	In Section \ref{sec:bergsma1} we give in brief the required background on this measure.
	
	Along with a similarity measure, we also need a spatial proximity matrix which reflects the strength of association based on the proximity between the regions. In Section \ref{sec:sb}, we use estimates of $\rho$,  and different choices for spatial proximities, to define a class of spatial association measures $S_B$. We have explored its asymptotic behavior, as the length of time increases,  assuming that the observations over time are i.i.d.~while at any fixed time, they are either pairwise independent or are dependent across the regions. We do not explore the asymptotic distributions under any possible additional spatio-temporal dependence. 
	
	In Section \ref{sec:simulation} we present the results of our simulation study for $S_B$ under both, the fully independent case, as well as in the case of two dependent models of spatial autoregressive and spatial moving average. In Section \ref{sb:covid} we present an application of the proposed measure of spatial association to COVID-19 incidence data of the Indian state of Kerala. This is done by fitting an appropriate model, and then using the $S_B$ statistic computed from the appropriate residuals.

\section{Methods and Materials}\label{methods}

\subsection{Spatial Association Statistics}\label{sec:spatialassoc}
	Suppose we  have a study area which is divided into $R$ (geographical) units over which a variable is observed. Spatial association refers to the relationship between the values of the variable with respect to the proximity of the regions. The two building blocks of a spatial association statistic are two matrices, $W$ and $Y$.

	The \textbf{spatial proximity matrix} $W=((w_{ij}))_{1\leq i, j \leq R}$ is based on the proximity (typically the geographical proximity) of the units, and provides the spatial component. Larger weights are assigned to the pairs of regions that are spatially ``more related". \cite{getis2010analysis} suggests three different types of $W$ matrices: 
	(i) a matrix based on some theoretical notion of spatial association, such as a distance decline function, 
	(ii) a matrix based on a geometric indicator of spatial nearness, such as the representation of contiguous
	spatial units, and 
	(iii) a matrix which uses a descriptive expression of the spatial association that already exists within the data set.
	It is always assumed that $w_{ii}=0$ for all $i$. 
	The most commonly used $W$ is the \textit{adjacency matrix} where, 
	\[
	w_{ij}= 
	\begin{cases}
		1& \text{if regions $i$ and $j$ are adjacent localities, } \\
		0              & \text{otherwise.}
	\end{cases}
	\]
	Another popular choice is to take $d_{ij}$ as the Euclidean distance, or some other distance, between the centroids of regions $i$ and $j$, and then $w_{ij}=d_{ij}^{-1}$.
	
	The matrix $W$ is often \textit{row-standardized}, so that each row sum in the matrix is equal to one. We would be working with $W$ which are row-standardized. This will turn out to be a crucial point in our data analysis later.
	
	The second component is a \textbf{similarity matrix}. Suppose that we have a series of observations ${X_i}^T =\{x_{im},m=1,\ldots, T\}$ at each region $i=1,2,\ldots,R$ at $T$ time points. The similarity matrix $Y$ is defined as $Y:=((sim_{ij}))$, where $sim_{ij}$ is some measure of similarity between $X_i^T$ and $X_j^T$. For example, it could be the sample covariance between $X_i^T$ and $X_j^T$. Note that these values are dependent only on any possible non-spatial relationship across the space. 
	
	Any choice $W$ and $Y$ yields a global \textbf{spatial autocorrelation index} 
	\begin{equation} 
		\dfrac{\displaystyle{\sum_{i=1}^{R}\sum_{j=1}^{R}w_{ij}sim_{ij}}}{\displaystyle{\sum_{i,j=1}^{R}w_{ij}}}.
		\label{basicf}
	\end{equation}
	When the matrices $Y$ and $W$ have similar structures, that is, they have high or low values together for the same pair of units $i$ and $j$, we can say that there is a high degree of spatial association.
	
	When the observations $X_i^T$ are univariate with single observation per spatial unit, then the well-known Moran's $I$ (\cite{moran1948interpretation}) uses the covariance as a similarity measure. On the other hand, Gearcy's $c$ (\cite{geary1954contiguity} uses squared differences as the similarity.

\subsection{A measure of independence}\label{sec:bergsma1} \cite{ bergsma2006new} introduced a correlation coefficient $\rho$, with a corresponding covariance $\kappa$ as a measure of independence. It is known that $\rho(X,Y)=0$ if and only of $X$ and $Y$ are independent. 
Here we briefly give the definitions of this correlation $\rho$ and its estimators. For convenience, we use the same notation as that of \cite{kappara2022assessing}. 

Let $Z_1, Z_2$ be i.i.d.~with distribution $F$ which has finite mean. Define 
\begin{eqnarray}g_F(z)&:=& \BE_F[ |z-Z|], \label{eq:g_F}\\
	g(F)&:=&\BE_{F}[|Z_1-Z_2|]=\BE_{F}[g_F(Z)], \label{eq:gF}\\
	h_F(z_1, z_2)&=&-\dfrac{1}{2} \big[|z_1-z_2|-g_F(z_1)-g_F(z_2)+g(F)\big].\label{eq:hdef}
\end{eqnarray}
Note that 
\begin{equation}\label{eq:iidcase}
	\mathbb{E}h_F(Z_1,Z_2)=0\ \ \text{whenever}\ \ Z_1, Z_2\ \ \text{are i.i.d.} \ \ F.
\end{equation}
Now, let $F_1$ and $F_2$ be the marginal distributions of a bivariate random variable $(X,Y)$. Let $(X_1,Y_1)$ and $(X_2,Y_2)$ be i.i.d~copies of $(X,Y)$. Then Bergsma's covariance and correlation coefficient  between $X$ and $Y$ are defined respectively by
\begin{equation}
\kappa=	\kappa(X,Y) :=\mathbb{E}[h_{F_{1}}(X_1,X_2)h_{F_{2}}(Y_1,Y_2)], \ \ \text{and}\ \ \rho=\rho(X,Y) :=\frac{\kappa(X,Y)}{\sqrt{\kappa(X,X)\kappa(Y,Y)}}. 
	\label{kappa}
\end{equation}

\subsection{Estimates of \texorpdfstring{$\kappa$}{}} Suppose we have $n$ observations $(x_i,y_i)$, $1\leq i \leq n$ from a bivariate distribution $F_{12}$. Bergsma has given two estimators of $\kappa$ namely $\tilde{\kappa}$, and $\hat{\kappa}$ respectively based on some $U$ and $V$-statistics with estimated kernels, and studied their distributional properties under independence. In \cite{kappara2022assessing} a third estimate $\kappa^{*}$ was introduced, and the properties of all three covariance estimates $\tilde{\kappa}$,  $\hat \kappa$ $\kappa^{*}$ was discussed under dependence as well as independence. In particular, their asymptotic normality was established under dependence.  

Simulation results in \cite{kappara2022assessing} show that the $V$-statistic based estimate $\hat{\kappa}$ has an upward bias. Performance-wise, $\tilde{\kappa}$ had an edge in computation time over the other two estimators. Therefore, we shall use  the $U$-statistic based estimate $\tilde{\kappa}$ and the corresponding correlation $\tilde{\rho}$  in our subsequent developments. 

Note that the kernel function $h_F$ defined in Equation \eqref{eq:hdef}, depends on the unknown distribution function $F$. 
Its sample analogue is given by  $\tilde{h}_{\hat{F}}$,
\begin{equation}
	\tilde{h}_{\hat{F}}(z_i,z_j)=-\frac{1}{2}\bigg[|z_i-z_j|-\frac{n}{n-1}(\frac{1}{n}\sum_{k=1}^{n}|z_i-z_k|+\frac{1}{n}\sum_{k=1}^{n}|z_k-z_j|-\frac{1}{n^2}\sum_{k,l=1}^{n}|z_k-z_l|)\bigg], \label{eq:hftilde}
\end{equation}
 where $\hat{F}$ is the sample distribution function. The $U$-statistic type estimator of $\kappa$ is defined as,
\begin{equation}
	\tilde{\kappa}=\tilde{\kappa} (x, y):=\binom{n}{2}^{-1}\sum_{1\leq i < j \leq n}\tilde{h}_{\hat{F}_1}(X_i, X_j)\tilde{h}_{\hat{F}_2}(Y_i, Y_j), \ \ \text{and}\ \ \tilde{\rho}=\tilde{\rho}(x,y):=\frac{\tilde{\kappa}(x,y)}{\sqrt{\tilde{\kappa}(x,x)\tilde{\kappa}(y,y)}},\label{rhou}
\end{equation}
where $\hat{F}_1$ and $\hat{F}_2$ are the sample distribution functions of the $\{X_i\}$ and $\{Y_i\}$ respectively. 

\cite{kappara2022assessing} proved that if the pairs $\{(X_i, Y_i)\}$ are i.i.d.~with $\BE_{F_{12}}[X_1^2Y_1^2] < \infty$, then  as $n \to \infty$, 	$n^{1/2}\big(\tilde{\kappa}-\kappa\big)$ is asymptotically normal with mean $0$ and some variance $\delta_1$. 
The crucial step in the proof of the above result which we shall need later, is that the leading term of $\tilde \kappa$ is a $U$-statistics, whose first projection is say $H_1$, and 
\begin{equation}\label{kappahathdecom}n^{1/2}\big(\tilde{\kappa}-\kappa\big)=\frac{1}{2}n^{-1/2}\sum_{i=1}^{n}H_1(X_i, Y_i)+R_n, \ \text{where}\ \  
R_n \xrightarrow{P} 0.\end{equation}
	To describe $H_1$ and $\delta_1$, let
 	\begin{eqnarray*} g_{F_{12}}(x,y)&:=&\BE_{F_{12}} \big[|x-X_1| \  |y-Y_1|\big], \\
		g({F_{12}}) &:=& \BE_{F_{12}} \big[ g_{F_{12}} (X_2, Y_2)\big]= \BE_{F_{12}} \big[|X_2-X_1| \  |Y_2-Y_1|\big].
	\end{eqnarray*}
$$\mu_{1}:=g(F_1), \ \mu_{2}:=g(F_2), \ \mu_{12}:=g (F_{12}), \ \text{and} \ \ \mu_{3}:=\BE_{F_{12}}\big[g_{F_{1}}(X)g_{F_{2}}(Y)\big].$$
	Let $(X_1, Y_1),  (X_2, Y_2), (X_3, Y_3)$ be i.i.d. Then $H_1(\cdot, \cdot)$ and $\delta_1$ are given by 
	\begin{eqnarray*}
		H_1(x, y)&=&\big[g_{F_{12}}(x, y)-\mu_{12}+\mu_{1}(g_{F_{2}}(y)-\mu_2)+\mu_{2}(g_{F_{1}}(x)-\mu_1)-(g_{F_{1}}(x)g_{F_{2}}(y))\\
		&& \ - \BE_{F_{12}}[|X_2-x|\ |Y_2-Y_3|  |]-\BE_{F_{12}}[|X_2-X_3|\ |Y_2-y| ]+3\mu_{3}\big],\\
	\delta_1 &=&\dfrac{1}{4}\mathbb{V} (H_1(X_1, Y_1)).\end{eqnarray*}
	

When $\kappa=0$, that is $X$ and $Y$ are independent, the first projection $H_1$ in \eqref{kappahathdecom} is zero. In this case, \cite{bergsma2006new} obtained the asymptotic distributions of $\tilde{\kappa}$ and $\hat{\kappa}$.
\cite{kappara2022assessing}, gave a detailed proof of the above result, along with the asymptotic distribution of $\kappa^*$. Again, the crucial step in that proof is to identify the leading term of the estimators, in terms of the second projections of the relevant $U$-statistics.
Indeed, \begin{equation}
	\tilde{\kappa}=\binom{n}{2}^{-1}\sum_{1\leq i<j \leq n}h_{F_1}(X_i, X_j)h_{F_2}(Y_i, Y_j)+{R}_n, \ \text{where}\ \ n{R}_n\xrightarrow{P}0. \label{hdecom2}
\end{equation}
Later we shall use this weak expansion to find the asymptotic distribution of our spatial association measure when the regions are assumed to be pairwise independent.

\section{The \texorpdfstring{$S_B$}{} measure of association}\label{sec:sb} Now suppose there are $R$ regions and let $X_i$ denote a variable corresponding to the region $i=1,\ldots, R$. Then we can use the similarity measure $\rho$, to define a global spatial association measure (\textbf{Spatial Bergsma}) as 
\begin{equation}
	S_{B}:=\dfrac{\displaystyle{\sum_{i, j=1}^{R}w_{ij}\rho (X_i, X_j)}}{S_0},\  \text{where}\ \ S_0=\sum_{i, j=1}^{R}w_{ij}.\label{sbp}
\end{equation} 
Note that 
\begin{equation}\label{eq:rowstan}
S_0=R, \ \text{whenever}\ \ W\ \ \text{is row-standardized}. 
\end{equation}
 Further,  $\rho (X_i, X_j)=\rho (X_j, X_i)$ and $w_{ii}=0$ for all $i,j=1,\ldots,R$. Therefore,
\begin{eqnarray}
	S_{B} = \dfrac{\displaystyle{\sum_{1\leq i <j \leq R}(w_{ij}+w_{ji})\rho(X_i, X_j)}}{S_0}.\label{SBp1}
\end{eqnarray}
Now, suppose that we have observations ${X_i}^T=\{x_{im},m=1,\ldots, T\},$ such that the vectors $(x_{im}, 1\leq i \leq R)$, are i.i.d.~for $m=1,\cdots, T$, with marginal distributions $F_i$, $1\leq i \leq R$. Let $\tilde{\kappa}^{(ij)}=\tilde{\kappa}(X_{i}^{T},X_{j}^{T})$, and $\tilde{\rho}^{(ij)}$ denote the $\tilde{\kappa}$ and $\tilde{\rho}$, calculated for the pair of variables $(X_i^{T},X_j^{T})$ using Equations \eqref{eq:hftilde} and \eqref{rhou}. 

We define a $U$-statistic based estimate of $S_{B}$ as
\begin{equation}
	\tilde{S}_{B}:=\dfrac{\displaystyle{\sum_{1\leq i <j \leq R}(w_{ij}+w_{ji})\tilde{\rho}^{(ij)}}}{S_0},\label{sb}
\end{equation} 
where,
\begin{equation}
	\tilde{\rho}^{(ij)}=\tilde{\rho}(X_{i}^{T},X_{j}^{T})=\frac{\tilde{\kappa}(X_{i}^{T},X_{j}^{T})}{\sqrt{\tilde{\kappa}(X_{i}^{T},X_{i}^{T})\tilde{\kappa}(X_{j}^{T},X_{j}^{T})}}. \label{bij}
\end{equation}
Further, with a given spatial proximity $((w_{ij}))$, we can compute this \textit{spatial Bergsma statistic} $\tilde{S}_{B}$ using Equation \eqref{sb}. The {\tt R} code for the computation of $\tilde{S}_{B}$ is given in the Appendix A.

\subsection{Asymptotic distribution of \texorpdfstring{$\tilde{S}_B$}{}}
The finite sample distributional properties of global indices are challenging to obtain in general. Usually additional assumption of normality is made for this purpose, or a randomization approach is used. When $T$ is large,  asymptotic distributions are also used. \cite{tiefelsdorf1995exact} showed that the accuracy of the asymptotic distribution of global indices depend heavily on the spatial proximity matrix $W$ and the number of regions $R$. 

Similar to the other global indices, it is impractical to obtain the exact distribution of $S_B$ for fixed $T$. Therefore we focus on  asymptotic results. Under appropriate assumptions, the asymptotic normality of $\tilde{S}_B$, as $T\to \infty$ can be easily established  as follows.

\noindent Writing Equation \eqref{sb} explicitly we get
\begin{equation}
	\tilde{S}_{B} = \frac{1}{S_0}\sum_{1\leq i <j \leq R}(w_{ij}+w_{ji})\frac{\tilde{\kappa}(X_i^{T},X_j^{T})}{\sqrt{\tilde{\kappa}(X_i^{T},X_i^{T})\tilde{\kappa}(X_j^{T},X_j^{T})}}.	\label{sbdef}
\end{equation}
Define an $\binom{R}{2} \times 1$ column vector $\mathbf{C}$, and a $1\times \binom{R}{2}$  row vector $\mathbf{d}$ as
\begin{eqnarray*}\mathbf{C}
	&:=&\big(\tilde{\kappa}^{(ij)},1\leq i<j\leq R\big)^{\top},\\
	\mathbf{d}
	&:=& \big(\frac{(w_{ij}+w_{ji})}{R\sqrt{\tilde{\kappa}(X_i^{T},X_i^{T})\tilde{\kappa}(X_j^{T},X_j^{T})}},1\leq i<j\leq R\big).\nonumber	\end{eqnarray*}
Then we can rewrite Equation \eqref{sbdef} as 	$$\tilde{S}_B=\mathbf{d}\mathbf{C}.$$
Define the centered and scaled column vector, 
$$\tilde{\mathbf{C}}:=\big(\sqrt{T}\big(\tilde{\kappa}^{(ij)}-{\kappa}^{(ij)}\big),1\leq i<j\leq R\big)^{\top}.$$
Define $\tilde{S}_B^{\prime}$ after required centering and scaling of $\tilde{S}_B$ as,
\begin{equation*}
	\tilde{S}^{'}_B:= \mathbf{d}\tilde{\mathbf{C}}.
\end{equation*}
Note that for every $i$, ${X_i}^T, T \geq 1$ are i.i.d. By an application of the SLLN it follows that, 
\begin{eqnarray} 
	\tilde{\kappa}(X_i^{T},X_i^{T}) & \xrightarrow {a.s} & \mathbb{E}h_{F_{i}}(X_i,X_i)^2 \nonumber\\
	&=&\sum_{k=0}^{\infty}\big(\lambda_k^{(i)}\big)^2. \label{denm}
\end{eqnarray}
Hence
\begin{equation} 
	\mathbf{d} \xrightarrow {a. s.} \big(\frac{(w_{ij}+w_{ji})}{S_0\sqrt{\displaystyle{\sum_{k=0}^{\infty}\big(\lambda_k^{(i)}\big)^2}\displaystyle{\sum_{k=0}^{\infty}\big(\lambda_k^{(j)}\big)^2}}},1\leq i<j\leq R\big)=\mathbf{a}\ \ \text{(say)}.
	\label{limitdenom}
\end{equation}
Recall that $S_0=R$ if $W$ is row-standardized. 

We are now ready to state the asymptotic normality result for the $S_B$ statistic. We use the notation $H_{1}^{(ij)}$ to denote the kernel function $H_1$ corresponding to the regions $(i, j)$. 
\begin{theorem}\label{thm:sbstatnormal}Suppose that the vectors
$\{(x_{im}, 1\leq i \leq R)\}, m=1,\cdots, T$ are i.i.d.~such that $\mathbb{V}({H_1}^{(ij)}(x_{im},x_{jm}))<\infty$, for all 
pairs $(i,j)$. Then as $T\rightarrow\infty$, 
	\begin{equation*}
		\tilde{S}^{\prime}_B= \mathbf{d}\tilde{\mathbf{C}} \xrightarrow{D} N(0,\mathbf{a^{\top}}\Sigma \mathbf{a}).
	\end{equation*}
where $\mathbf{a}$ is as in (\ref{limitdenom}) and  the covariance matrix $\Sigma$ is defined in the proof given below. \end{theorem}
\begin{proof}
	Note that 
	\begin{eqnarray}
		\tilde{\kappa}^{(ij)}&=&\tilde{\kappa}(X_i^{T},X_j^{T})\\
		&=&\binom{T}{2}^{-1}\sum_{1\leq m < n\leq T}\tilde{h}_{\hat{F}_i}(x_{im},x_{in})\tilde{h}_{\hat{F}_j}(x_{jm},x_{jn}).
	\end{eqnarray}
	
	For any pair of regions $(i,j)$, by arguments given in Equation \eqref{kappahathdecom}
	\begin{equation}\label{hdecom1}
		\sqrt{T}(\tilde{\kappa}^{(ij)}-{\kappa}^{(ij)})=\frac{1}{2}T^{-1/2}\sum_{m=1}^{T}{H_1}^{(ij)}\big(x_{im},x_{jm}\big)+{R}^{(ij)}_T,\ 
		\text{where}\ \ {R}^{(ij)}_T\xrightarrow{P} 0. 
	\end{equation}
	Hence by the multivariate central limit theorem, 
	\begin{equation}
		\tilde{\mathbf{C}}_{\binom{R}{2}\times 1}=\big(\sqrt{T}\big(\tilde{\kappa}^{(ij)}-{\kappa}^{(ij)}, 1\leq i < j \leq R\big) \xrightarrow{D} N\big(\mathbf{0},\Sigma_{\binom{R}{2} \times \binom{R}{2}}\big), \label{numsb}
	\end{equation}
	where, $\Sigma$ is the covariance matrix $\Sigma=((\sigma_{i_1j_1,i_2j_2}))_{1\leq i_1<j_1\leq R, 1\leq i_2<j_2\leq R}$ with 
		\begin{equation}
		\sigma_{i_1j_1,i_2j_2}=\frac{1}{4}\mathbb{COV}\big({H_1}^{(i_1j_1)}\big(x_{i_1m},x_{j_1m}\big),{H_1}^{(i_2j_2)}\big(x_{i_2m},x_{j_2m}\big)\big).\end{equation}
	For $i_1=i_2=i$ and $j_1=j_2=j$,
	\begin{equation}\label{eq:delta1}
		\sigma_{i_1j_1,i_2j_2}=\frac{1}{4} \mathbb{V}\big({H_1}^{(ij)}\big(x_{im},x_{jm}\big)\big)={\delta_1}^{(ij)}.
	\end{equation}
	Therefore by Equations \eqref{numsb} and \eqref{limitdenom}, 
		$\tilde{S}^{'}_B= \mathbf{d}\tilde{\mathbf{C}} \xrightarrow{D} N(0,\mathbf{a^{\top}}\Sigma \mathbf{a})$,
	completing the proof.
\end{proof}

\subsection{Asymptotic distribution of \texorpdfstring{$S_B$}{} under spatial pairwise independence}\label{sec:sbasymnull}
Now, suppose that the regions are pairwise independent, that is, $\tilde{\kappa}^{(ij)}=0$ for all $1\leq i < j \leq R$. Then it is known that $\delta_1^{ij}$ given in (\ref{eq:delta1}) is $0$ for all $i\neq j$, and the limit distribution in the above result is degenerate at $0$. 
In this case we have the following asymptotic distribution result.
\begin{theorem}\label{theo:independence}
	Suppose that the vectors $(x_{im}, 1\leq i \leq R)$ are i.i.d.~for $1\leq m \leq T$. Further, for any $1\leq i\neq j\leq R$, $x_{im}$ and $x_{jm}$ are independent. 	Let $F_i$ be the marginal distribution of $x_{im}$. Suppose that $h_{F_i}$ is square integrable with the eigen decomposition (in the $L^2$ sense),
	\begin{equation}
		h_{F_i}(x, y)=\sum_{k=1}^{\infty}\lambda_k^{(i)}g_k^{(i)}(x)g_k^{(j)}(y). \label{SBhF}
	\end{equation}
	Then as $T\rightarrow\infty$,
	\begin{equation}
		T\tilde{S}_{B}\xrightarrow D \frac{1}{S_0}\sum_{1\leq i <j\leq R}\big[\dfrac{(w_{ij}+w_{ji})
		\displaystyle{\sum_{k,l=0}^{\infty}\lambda_k^{(i)}\lambda_l^{(j)}(Z_{ik,jl}^{2}-1)}}{\sqrt{\displaystyle{\sum_{k,l=0}^{\infty}
				{\lambda_k^{(i)}}^{2}{\lambda_l^{(j)}}^{2}}}}\big],
	\end{equation}
	where $\{Z_{ik,jl}\}$, are i.i.d.~standard normal variables, and $S_0=R$ if $W$ is row-standardized.
\end{theorem}
\begin{proof}
Recall from \eqref{sbdef} that,
	\begin{equation}
		\tilde{S}_{B} = \frac{1}{S_0}\sum_{1\leq i < j \leq R}(w_{ij}+w_{ji})\frac{\tilde{\kappa}(X_i^{T},X_j^{T})}{\sqrt{\tilde{\kappa}(X_i^{T},X_i^{T})\tilde{\kappa}(X_j^{T},X_j^{T})}}.	 \label{nullproof}
	\end{equation}
	Further we have
	\begin{equation}
		\tilde{\kappa}^{(ij)}=\tilde{\kappa}(X_i^{T},X_j^{T})=\binom{T}{2}^{-1}\sum_{1\leq m<n \leq T}\tilde{h}_{\hat{F}_i}(x_{im},x_{in})\tilde{h}_{\hat{F}_j}(x_{jm},x_{jn}). \label{kappaij}
	\end{equation}
As mentioned earlier, due to independence, the first projections $H_1^{(ij)}$ in \eqref{hdecom1} are zero. We now use the second projections from Equation \eqref{hdecom2}. 
	We can write,
	\begin{equation*}
		\tilde{\kappa}^{(ij)}=\binom{T}{2}^{-1}\sum_{1\leq m<n \leq T}{h}_{{F_i}}(x_{im},x_{in}){h}_{{F_j}}(x_{jm},x_{jn})+{R}^{(ij)}_T,
		\ \text{where}\ \ T{R}^{(ij)}_T \xrightarrow{P} 0.
	\label{H2decom}
	\end{equation*}
	
	Now we can follow the proof of the general theorem on the asymptotic distribution of a $U$-statistic with a general degenerate kernel. The difference is that now we have several degenerate $U$-statistics. Each of them converges by the limit theorem for degenerate $U$-statistics. See for example \cite{bose2018u}. We need to ensure the joint convergence of these statistics, and argue the independence of the 
$\{Z_{ik,jl}\}$ across $i, j$. We outline the arguments below. 

Using \eqref{H2decom}, we can write the following approximate equation for $\tilde{\kappa}$ calculated for any pair of regions, i.e., $\tilde{\kappa}^{(ij)}$ as follows:
	\begin{eqnarray*}
		T\tilde{\kappa}^{(ij)}
		&\simeq&\frac{2}{T}\big[\sum_{1\leq m<n \leq T}{h}_{{F_i}}(x_{im},x_{in}){h}_{{F_j}}(x_{jm},x_{jn})\big]\\
		&=&\frac{2}{T}\big[\sum_{1\leq m<n \leq T}\big[\sum_{k=0}^{\infty}\lambda_k^{(i)}g_k^{(i)}(x_{im})g_k^{(i)}(x_{in})\big]
		\big[\sum_{k=0}^{\infty}\lambda_k^{(j)}g_k^{(j)}(x_{jm})g_k^{(j)}(x_{jn})\big]\big]\\
		&=&\frac{1}{T}\sum_{m,n=1}^{T}\big[\sum_{k=0}^{\infty}{\lambda_{k}}^{(i)}g_k^{(i)}(x_{im})g_k^{(i)}(x_{in})\big]\big[\sum_{k=0}^{\infty}{\lambda_{k}}^{(j)}g_k^{(j)}(x_{jm})g_k^{(j)}(x_{jn})\big]\\ && \  -\frac{1}{T}\sum_{m=1}^{T}\big[\sum_{k=0}^{\infty}{\lambda_{k}}^{(i)}g_k^{(i)}(x_{im})g_k^{(i)}(x_{im})\big]\big[\sum_{k=0}^{\infty}{\lambda_{k}}^{(j)}g_k^{(j)}(x_{jm})g_k^{(j)}(x_{jm})\big]\\
		&=& T_1 - T_2 \ \ (say).
	\end{eqnarray*}
Note that $\{g_{k}^{(i)}(\cdot)\}$ and $\{g_{l}^{(j)}(\cdot)\}$ are orthonormal functions. We explain in brief the convergence of the two terms. The second term equals
\begin{eqnarray*}
	T_2&=&\frac{1}{T}\sum_{m=1}^{T}\big[\sum_{k=0}^{\infty}{\lambda_{k}}^{(i)}g_k^{(i)}(x_{im})g_k^{(i)}(x_{im})\big]\big[\sum_{l=0}^{\infty}{\lambda_{l}}^{(j)}g_l^{(j)}(x_{jm})g_l^{(j)}(x_{jm})\big]\\
	& \xrightarrow{a.s.} & \big(\sum_{k=0}^{\infty}{\lambda_{k}}^{(i)}\big) \big(\sum_{l=0}^{\infty}{\lambda_{l}}^{(i)}\big).
\end{eqnarray*}
The first term equals, 
\begin{eqnarray*}
	T_1&=&\frac{1}{T}\sum_{m,n=1}^{T}\big[\sum_{k=0}^{\infty}{\lambda_{k}}^{(i)}g_k^{(i)}(x_{im})g_k^{(i)}(x_{in})\big]\big[\sum_{l=0}^{\infty}{\lambda_{l}}^{(j)}g_l^{(j)}(x_{jm})g_l^{(j)}(x_{jn})\big]\\
	&=&	\sum_{k,l=1}^{\infty}{\lambda_{k}}^{(i)}{\lambda_{l}}^{(j)}\big[\frac{1}{T}\sum_{m,n=1}^{T}g_k^{(i)}(x_{im})g_k^{(i)}(x_{in})g_l^{(j)}(x_{jm})g_l^{(j)}(x_{jn})\big]\\
	&=&	\sum_{k,l=1}^{\infty}{\lambda_{k}}^{(i)}{\lambda_{l}}^{(j)}\big[(\frac{1}{\sqrt{T}}\sum_{m=1}^{T}g_k^{(i)}(x_{im})g_l^{(j)}(x_{jm}))(\frac{1}{\sqrt{T}}\sum_{n=1}^{T}g_k^{(i)}(x_{in})g_l^{(j)}(x_{jn}))\big]\\
	&\xrightarrow{D}& \sum_{k,l=0}^{\infty}{\lambda_{k}}^{(i)}{\lambda_{l}}^{(j)}Z_{ik,jl}^{2}.
\end{eqnarray*}
Therefore we have, for every pair $i < j$,  	
$$T\tilde{\kappa}^{(ij)} \xrightarrow{D} \sum_{k,l=0}^{\infty}\lambda_k^{(i)}\lambda_l^{(j)}(Z_{ik,jl}^{2}-1),$$
where $Z_{ik,jl}, 1\leq k, l < \infty$ are independent standard normal variables. 

Moreover, when we consider the joint convergence of $T\tilde{\kappa}^{(ij)}, 1\leq i < j\leq R$, due to the pairwise independence $X_i^{T}$ and $X_j^{T}$, and the orthonormality of $\{g_{k}^{(i)}(\cdot)\}$ and $\{g_{l}^{(j)}(\cdot)\}$, the variables $Z_{ik,jl}$ are all  independent of each other.  Incidentally, only pairwise independence is being used here.
That is, 
	\begin{equation}
		\big(T{\kappa}^{(ij)},1\leq i < j \leq R\big) \xrightarrow D \big(\sum_{k,l=0}^{\infty}\lambda_k^{(i)}\lambda_l^{(j)}(Z_{ik,jl}^{2}-1),1\leq i < j \leq R\big). \label{numvec}
	\end{equation}

	Now, in the denominator of \eqref{nullproof} we have terms of the form $\tilde{\kappa}(X_i^{T},X_i^{T})$ that converge to constants given in Equation \eqref{denm}.
	
	From Equations \eqref{numvec} and \eqref{denm} we obtain that,
	\begin{equation}
		T\tilde{S}_{B}\xrightarrow D \frac{1}{S_0}\sum_{1\leq i < j\leq R}\big[\frac{(w_{ij}+w_{ji})
		\displaystyle{\sum_{k,l=0}^{\infty}\lambda_k^{(i)}\lambda_l^{(j)}(Z_{ik,jl}^{2}-1)}}{\sqrt{\displaystyle{\sum_{k=0}^{\infty}
				\big(\lambda_k^{(i)}\big)^{2}}\displaystyle{\sum_{l=0}^{\infty}\big(\lambda_l^{(j)}\big)^{2}}}}\big], \label{NULL}
	\end{equation}
	completing the proof.
\end{proof}
It may be noted that the limit distribution of $S_B$ given in Equation \eqref{NULL} depends on 
$\{F_i\}$ 
through the eigenvalues $\{\lambda_k^{(i)}\}$
given by Equation \eqref{SBhF}. 
For most distributions these cannot be computed in a closed form. However, they can be 
numerically approximated.  

Consider the special case where $\{X_{i1}, i=1,\ldots,R\}$ have a common distribution $F$, so that 
$F_i \equiv F$ for all $i$. Let $(\lambda_{k}^{(i)},g_{k}^{(i)})\equiv(\lambda_{k},g_k)$ for all $1\leq i\leq R$.
Then 
\begin{equation}
	T\tilde{S}_{B}\xrightarrow D \frac{1}{S_0\displaystyle{\sum_{k=0}^{\infty}
		\lambda_k^{2}}}\sum_{1\leq i < j \leq R}\big[{(w_{ij}+w_{ji})\sum_{k,l=0}^{\infty}\lambda_k\lambda_l(Z_{ik,jl}^{2}-1)}\big]. \label{NULL1}
\end{equation}


\section{Simulation study}\label{sec:simulation}

Next we explore various distributional aspects of the $\tilde{S}_B$ statistic through simulations. Our choice of $R$ and $W$ are motivated by the district wise COVID-19 data from the state of Kerala in India, that we are going to explore in the next section. 

Thus, we use $R=14$ as there are 14 districts of Kerala. We use three different spatial proximity matrices:
(i) the lag-1 adjacency matrix of the districts (i.e. regions) 
(ii) distance matrix of these districts, using the average of latitude and longitude of the district headquarters, and (iii) linear connectivity matrix, which is a lag-1 adjacency matrix with the regions arranged in a line. The third $W$ matrix is motivated by the almost linear geographical organization of the districts of Kerala from North-West to South-East.
We assume that we have observations over a reasonable length of time, and we choose $T=50$.

The $\tilde{S}_B$ statistic is studied by simulating it under various scenarios in absence and presence of spatial association. This requires obtaining samples so that we can apply formulae \eqref{sbp} and \eqref{SBp1}. Note that $S_0=R$ since we are working with row-standardized $W$ matrices. 

First we study the distributions under the null case of no dependence between the regions, and also validate the asymptotic approximations obtained in \eqref{NULL} and \eqref{NULL1}. Note that to use the finite-discrete approximation of the asymptotic distribution, we would require the eigenvalues of the kernels for the corresponding distributions.


Then we further study the $\tilde{S}_B$ statistic in presence of spatial dependence, using two well-known spatial processes, namely spatial autoregression and spatial moving average.

\subsection{Null Case of spatial independence}
In this simulation exercise we have used a common choice of the distribution $F$ for all $R$ spatial locations--six different choices for $F$, namely normal, uniform, exponential, Laplace, logistic, and chi-square were explored.  The eigenvalues $\{\lambda_k\}$ are readily available for these cases.  
Empirical null distributions were obtained using 10,000 replicates of the estimates of $\tilde{S}_B$, under spatial pairwise independence and these six choices of $F$.

From Figure \ref{fig:nullsims} we observe that under the null scenario of spatial independence, the finite sample distribution of the $\tilde{S}_B$ estimate does not depend in any significant way on 
the distributional assumptions. However there are visible effects of the spatial proximity matrices.

\begin{figure}[ht]
\centering
\includegraphics[width=1\linewidth]{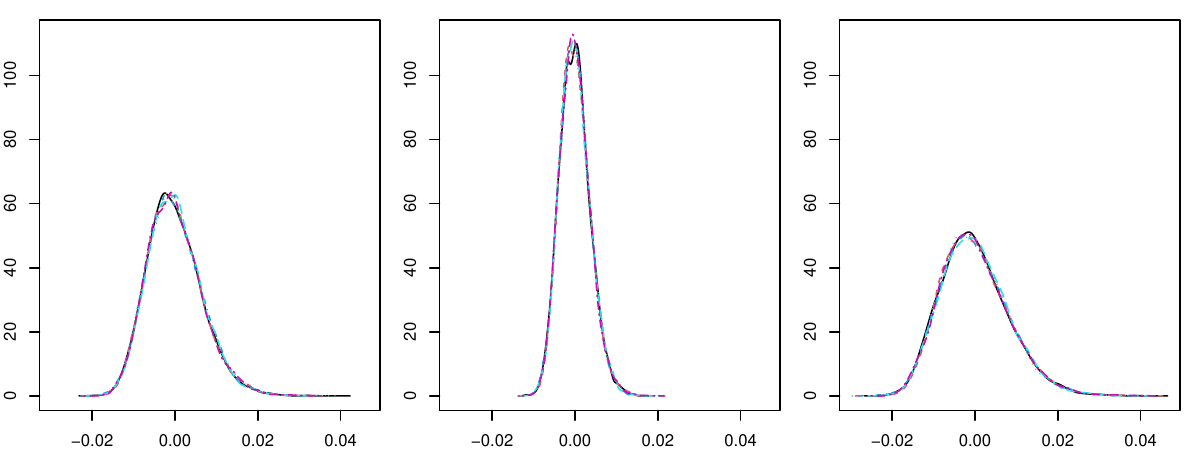}
\caption{Simulated null distribution of $T\tilde S_B$, $T=50$, $R=14$ for $10,000$ replicates, for six distributions. Proximity matrices $W$: lag-$1$ adjacency (left), 
inverse distance (middle), and linearly arranged spatial units (right). }
\label{fig:nullsims}
\end{figure}

Next we validate the 
asymptotic behaviour of $\tilde{S}_B$ 
stated in Theorem \ref{theo:independence}. As mentioned earlier this would require discrete approximations using eigenvalues. For the detailed procedure of this discrete approximation method, see \cite{kappara2022assessing}. We consider the finite sum based on 100 eigenvalues to approximate the asymptotic null distributions.

This theoretical asymptotic distribution is computed for the case where all regions have identical distribution 
(normal in this illustration). This is then compared with the simulated empirical distribution generated under the same assumption of normality for all regions.

For each choice of spatial proximity matrix, we have approximated the null distributions, as described in \eqref{NULL1}, based on $10000$ samples, see Figure \ref{fig:nullcomprsns}.

Figure \ref{fig:nullcomprsns} illustrates  that even with only $T=50$, and $100$ eigenvalues, the asymptotic distribution provides a good approximation to the (empirical) null distribution. Also, as before, the underlying distributional assumption appears to be less critical than the choice of the spatial proximity matrix.
\begin{figure}[ht]
\centering
\includegraphics[width=1\linewidth]{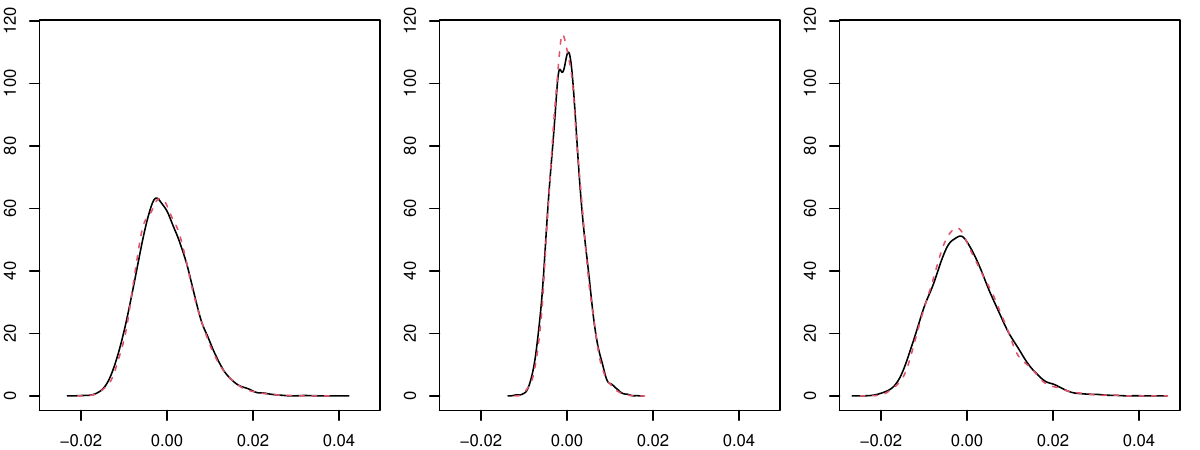}
\caption{Comparison of simulated null (black-solid line) and asymptotic null distributions (red-dashed-line) of $T\tilde S_B$ with $10,000$ replicates, based on finite discrete approximation using 100 eigenvalues for aymptotic distribution, with $W$ matrices: lag-$1$ adjacency (left), inverse distance (middle), and linearly arranged spatial units (right). }
\label{fig:nullcomprsns}
\end{figure}

\subsection{Spatial dependence case}
In order to illustrate the behaviour of the $\tilde{S}_B$ statistic under the presence of spatial dependence, we consider the following two general forms of spatial dependence models to simulate spatially dependent data. In simulations, we maintain the assumptions on the data given in Theorem \ref{thm:sbstatnormal}, that the observations are i.i.d.~across the time points.
\begin{enumerate}

\item Spatial moving average (SMA) dependence:
\begin{equation}
y=(I+\theta W)\epsilon. \label{SMA}
\end{equation}
\item Spatial autoregressive (SAR) dependence:
\begin{equation}
y=(I-\theta W)^{-1}\epsilon, \label{SAR}
\end{equation}
\end{enumerate}
where, $\theta$ is the spatial dependence parameter, $W$ (row standardized) corresponds to the assumed spatial proximity,   and $\epsilon$ is a vector of standard normal variates throughout the simulations. Under both the models,  $\theta=0$ refers to the null case of no spatial association. Note that \cite{anselin2012new} gives collection of some other spatial dependence models of higher orders.  

We consider a range of values for the spatial dependence parameter $\theta$, namely  0, 0.1, 0.25, 0.5, 0.75, and 0.9 in this illustration. We continue to use the same three choices of spatial proximity matrices described earlier, namely lag-1, inverse distance and linear connectivity. As before we have simulated spatially autocorrelated vectors of length $T=50$ at $R=14$ locations.  

For each choice of the $W$ matrix and the dependence parameter $\theta$, under the two spatial dependence models the $\tilde{S}_B$ statistic was simulated 10,000 times. In Figure \(\ref{fig:boxplots}\) we show the distribution of $\tilde{S}_B$ under SMA and SAR models.

\begin{figure}[ht]
\centering
\includegraphics[width=0.8\linewidth]{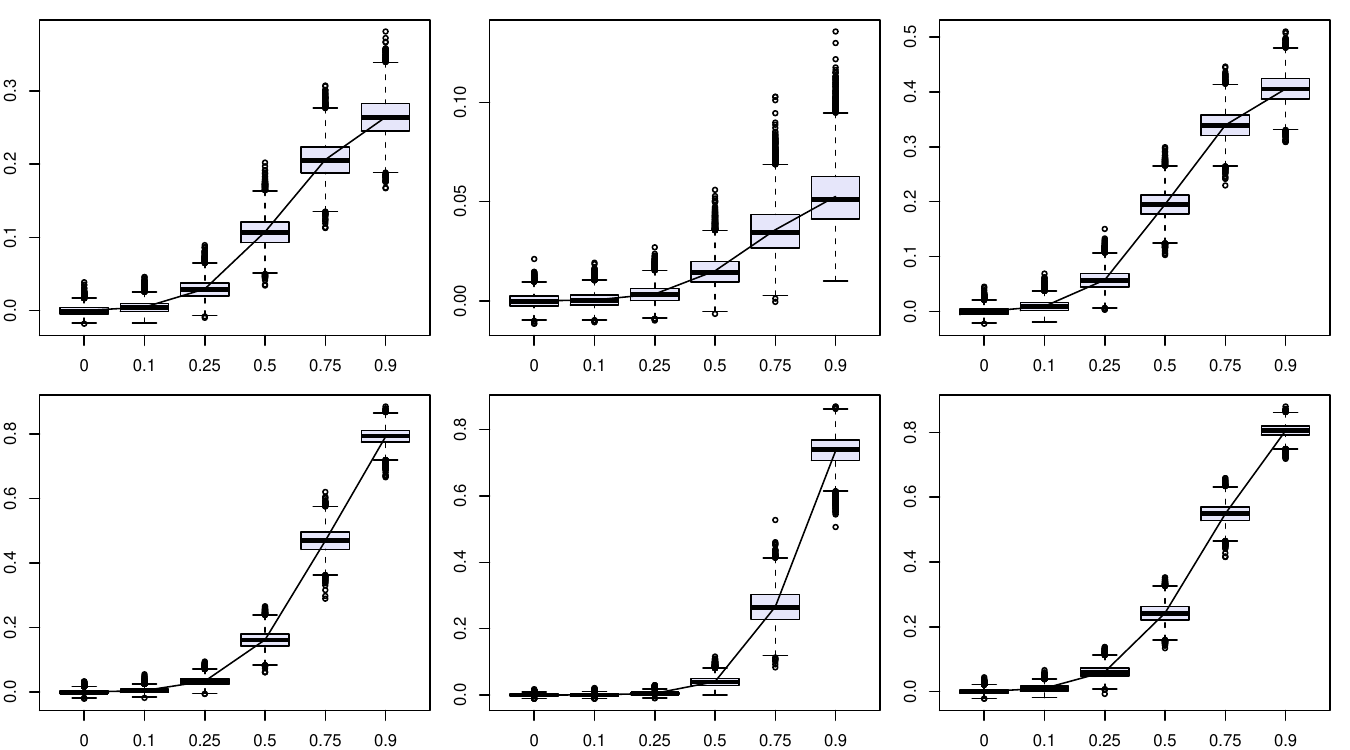}
\caption{Empirical distributions of $\tilde S_{B}^{'}$ (standardized version of $\tilde S_B$), $R$=14, $T$=50, $10,000$ replicates, with sample mean overlayed, dependence parameter $\theta$ on $x$-axis under SMA model (top), SAR model (bottom). Proximity matrices: lag-$1$ adjacency (left), inverse distance (middle), and linearly arranged spatial units (right).}
\label{fig:boxplots}
\end{figure}

The box plots in Figure \ref{fig:boxplots} show the distinct and monotone departure of the $\tilde{S}_B$ statistic from its null behaviour under the presence of spatial dependence. Incidentally, row-standardization of $W$ is crucial to yield this monotonicity.

It appears that as spatial association increases, the distribution seems to become more symmetric around the median. To confirm this, we study the skewness and kurtosis of this measure in Figure \ref{fig:skewkur}.
\begin{figure}[ht]
\centering
\includegraphics[width=0.5\linewidth]{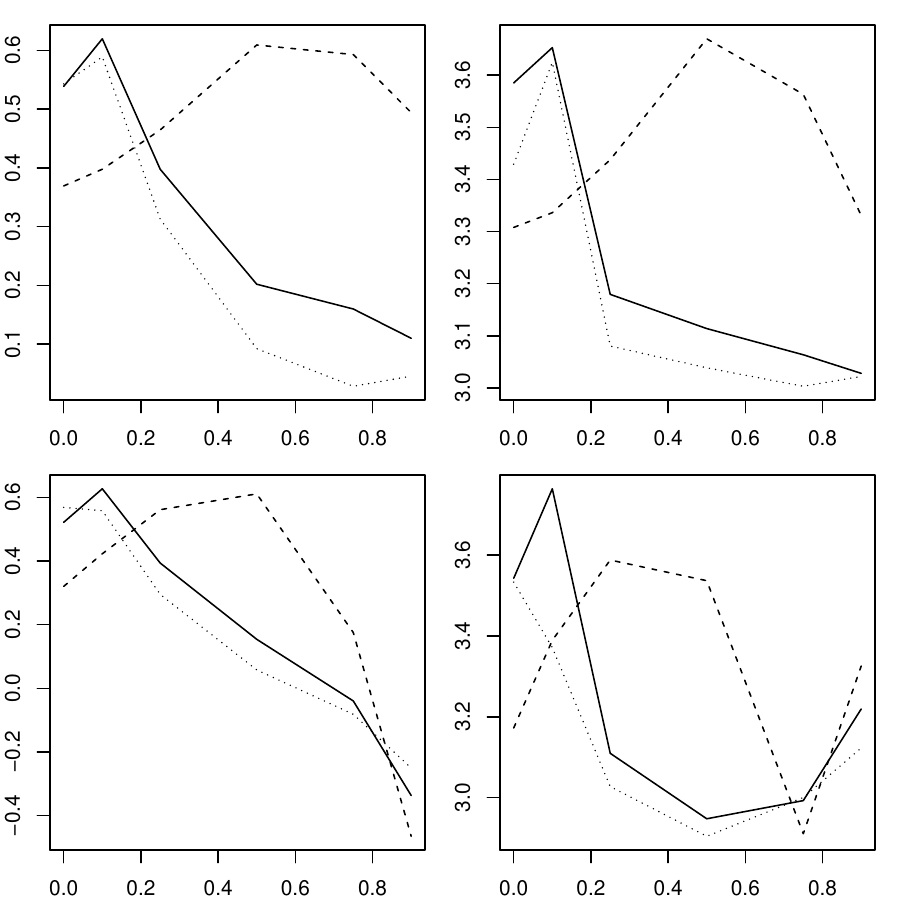}
\caption{Empirical skewness (left) and kurtosis (right) plotted against spatial dependence parameter $\theta$, under SMA model (top) and SAR model (bottom),  with proximity matrix: lag-$1$ adjacency matrix (black), inverse distance matrix (red), and linearly arranged spatial units (green) ($R$=14, $T$=50, $10,000$ replicates).}
\label{fig:skewkur}
\end{figure}
We see that for $\theta$ values close to $0$, i.e., relatively close to the null, the distribution of $\tilde{S}_B$ is moderately skewed. This observation is consistent with our result in Theorem \ref{theo:independence} that under the null, $\tilde{S}_B$ approximately has a distribution of a centered sum of weighted chi-squares.

As the spatial parameter $\theta$ increases, the skewness decreases gradually and kurtosis decreases and becomes close to $3$ under lag-1 adjacency and linear connectivity matrices (i.e. $W_1$ and $W_3$). This is consistent with the asymptotic normality result we have obtained in Theorem \ref{thm:sbstatnormal}.

Note that for lag-1 adjacency and the linearly arranged proximity matrices, the proportion of sparsity are $0.77$ and $0.87$. 
In comparison, the inverse distance $W$ matrix has a sparsity of $0.07$ only. Thus it is possible that the sparsity of the $W$ matrix also influences the distribution of $\tilde{S}_B$.

Overall we notice, from the Figures \ref{fig:boxplots} and \ref{fig:skewkur}, that the presence of spatial association highly influences the first four moments of the distribution of $\tilde{S}_B$.

\section{COVID-19 Application}\label{sb:covid} 

In this section we present an application of $\tilde{S}_B$ to the COVID-19 data from the Indian state of Kerala in the southern part of India. We consider district-level monthly COVID-19 data for the 14 districts of Kerala as of October 31, 2022. The data is obtained from a crowd-sourced database \url{https://data.incovid19.org}. We have the number of districts $R=14$ and the number of time points $T=29$ (months). 
Figure \ref{Districtts} presents the time series of this COVID-19 incidence data.

\begin{figure}[ht]
\centering
\includegraphics[width=0.9\linewidth]{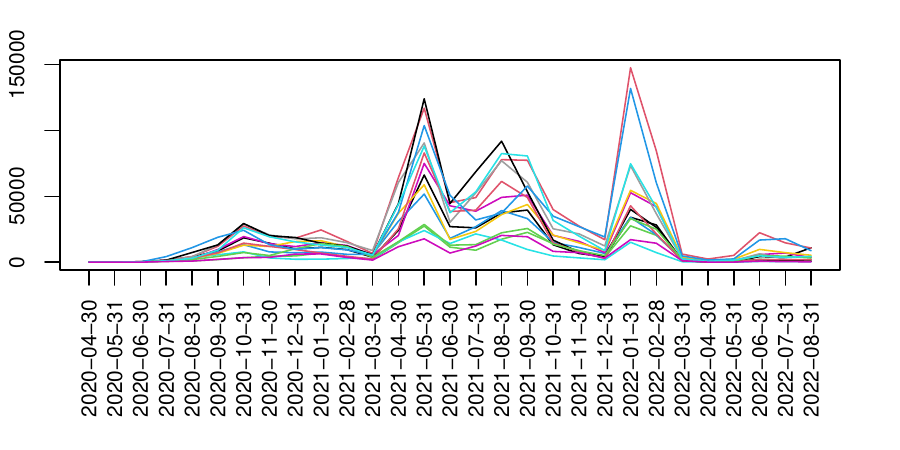}
\caption{District level COVID-19 time series data for Kerala }
\label{Districtts}
\end{figure}

We present the analyses here using the lag-1 adjacency matrix motivated by the works of \cite{Bhattacharjeeetal}.
In its current form, the $\tilde{S}_B$ statistics measures spatial association in absence of temporal pattern. However COVID-19 data naturally has temporal dependence. In order to meaningfully apply this measure, first  appropriate temporal models are applied to the data, and then the residuals are considered for assessing spatial dependence.



Apart from measuring spatial dependence, we would also assess the significance of such an $\tilde{S}_B$ measure. For that we use the asymptotic null distributions presented in Section \ref{sec:sbasymnull} and obtain an empirical p-value of the corresponding $\tilde{S}_B$ value. Further using bootstrap technique we also present confidence intervals.

From Figure \ref{Districtts} we observe a strong temporal pattern in the Covid-19 incidence.
However, it is also apparent that there is a (spatial) consistency in this temporal behaviour.

To capture these features, we explore three models, with the first one being only temporal, and other two involving spatio-temporal modelling. By nature, the first two models are endogenous (viz. temporal AR model and spatio-temporal autoregressive model), and the third one is exogenous (viz. a spatio-temporal-gravity model).

Let $x_{it}$=number of new cases for the $i$th district during the time period $t=1, \ldots, 29 (=T)$, $i=1, \ldots, 14 (=R)$. Then the models are defined as follows.\vskip5pt

\noindent \textbf{AR(3) Model.} Here we fit the AR(3) models to each of the $14$ district level time series separately and extract the residuals from each series. The model is defined as,
\begin{equation}
x_{it}=\beta^{(i)}_{0}+\beta^{(i)}_{1} x_{i(t-1)}+\beta^{(i)}_{2} x_{i(t-2)}+\beta^{(i)}_{3} x_{i(t-3)}+\epsilon_{it}.\label{AR}
\end{equation}
Here the superscript $(i)$ in the parameters $\beta$ represent parameters of the AR model fitted to the time series from the $i$-th district, $i=1, \ldots, 14(=R)$.\vskip5pt

\noindent \textbf{Spatio-temporal Model-1.} In this model we assume that for a given region $i$, the count for the $t$th month depend on the count for the $(t-1)$th month of that region along with all the $N_i$ neighbors of that region. That is, we regress each district's time series on its own past at lag-1 and the past of its  spatial neighbors. Thus the model involves both temporal and spatial covariates.

We further assume that the count response variable follows either a Poisson distribution, or a Negative Binomial distribution with mean $\mu_{it}$. Then the linear predictor for log-mean has the following form with the spatial and temporal autoregressive terms;
\begin{equation}
\log (\mu_{it})=\beta^{(i)}_{0}+\beta^{(i)}_{1} x_{i(t-1)}+\sum_{j \in N_i}\beta^{(i)}_{{2,j}}w_{ij}x_{j(t-1)},\hspace{0.5cm}t=1,\ldots,29,\hspace{0.5cm}j=1,\ldots,14. \label{SpAR}
\end{equation}
Again, the superscript $(i)$, for the $\beta$ parameters, indicates that the model is for the time series from the $i$th district, $i=1, \ldots, 14 (=R)$. Here too the models are fitted to the time-series from individual regions and residuals are extracted.
\vskip5pt

\noindent \textbf{Spatio-temporal Model-2.} The third model is an application of a specialized spatio-temporal model proposed in \cite{Bhattacharjeeetal}. All covariates in this model are exogenous. They include spatial-variates, temporal-variates, and covariates such as air passenger traffic data, that includes both spatial and temporal information. This model is fitted simultaneously to all $R$ time series due to shared parameters/covariates in the model.

The state of Kerala has four operational airports at Trivandrum, Calicut, Cochin and Kannur. The data on passenger traffic at these airports have been obtained from the Airport Authority of India (AAI) source \url{https://www.aai.aero/en/business-opportunities/aai}-\url{traffic-news}. We will denote the covariates for this model as follows: \vskip5pt

\noindent $d_{ik}$ is the (euclidean) distance between the $i$-th district center and the $k$-th  airport.\vskip5pt

\noindent $X_{1kt}$ is the number of passengers arriving at the $k$-th airport at $t$-th time point.\vskip5pt

\noindent $X_{2}$ and $X_{3}$ are two numerical variables containing geographical location of each district head quarter (i.e. longitude and latitude respectively).\vskip5pt

\noindent $X_4$ is a categorical variable for the $T$ time points.\vskip5pt

Once again we assume the response variable $x_{it}$ of COVID-19 incidence follows a negative binomial distribution, where the linear predictor for log-mean is given by,
\begin{equation}
\log\mu_{it}=\beta_0 + \sum_{k=1}^{4} \beta_{1,k} f(d_{ik})X_{1kt} +\beta_{2}X_{2}+\beta_{3}X_{3}+ \sum_{t=1}^{T} \beta_{4,t} I_{\{X_4=t\}} . \label{model}
\end{equation}
The first part of the model with parameters $\beta_{1,k}$ is adopted in the spirit of gravity models and explains the effect of the volume of air passenger traffic on the district level COVID-19 incidence. The (inverse) effect distance between airports and district headquarters is captured by a distance decay function. Based on the results from \cite{Bhattacharjeeetal}, we use two decay function, the inverse decay $f(d)=1/d$, and exponential decay $f(d)=\exp(-d)$. 
\begin{figure}[ht]
\centering
\includegraphics[width=0.7\linewidth]{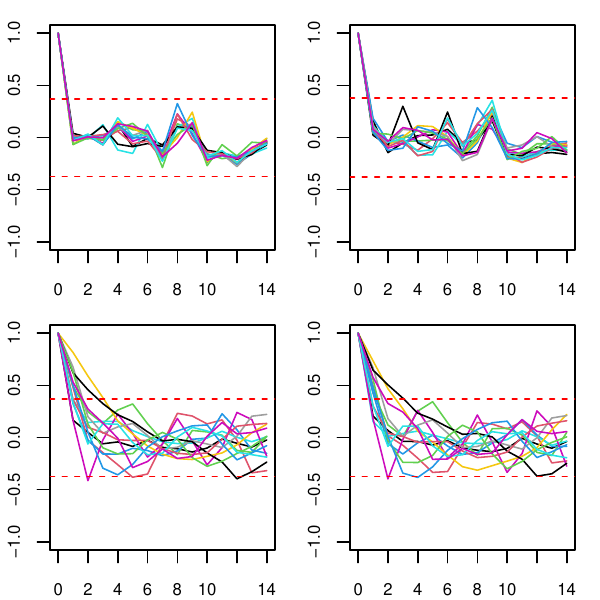}
\caption{Autocorrelation of residuals: Model \ref{AR} (top left), Model \ref{SpAR} (top right), Model \ref{model} with inverse distance decay (bottom left), and exponential decay (bottom right)  with $95\%$ threshold. }
\label{fig:acfplots}
\end{figure}

In Figure \ref{fig:acfplots} we present the autocorrelation functions from the 14 districts under these three models (with two distance decay functions for the third model).
From the residual autocorrelations plot from AR(3) model in Figure \ref{fig:acfplots} we observe that, in spite of fitting the models seperately for the 14 districts, the residuals behave similarly. This strengthens the idea of applying spatio-temporal models to this data.

From Figure \ref{fig:acfplots}), we also observe that the strong clustered pattern of the autocorrelation functions seem to dissipate with applications of the two spatio-temporal models. It appears to be most de-clustered for the spatio-temporal-model-2.

All four sets of autocorrelation plots in Figure \ref{fig:acfplots} are within the $95\%$ threshold for all $R$ time-series-residuals. Therefore for practical purposes we can assume an i.i.d.~structure within each of the residual series from the $R=14$ districts. Thus the statistic, $\tilde{S}_{B}$ can be computed based on these residuals, and we can apply our distributional results. 

The above observations are further confirmed by carrying out a test of significance for the $\tilde{S}_B$ statistics (see Table \ref{table:1}).
\begin{table}[ht]
\caption{\label{table:1} Results on $\tilde{S}_B$ from the fitted models for COVID-19 incidence data}
\vskip3pt
\begin{center}
\begin{tabular}{l|c|c|c}
\hline
\textbf{Model} & $\tilde{S}_B$ & \textbf{CI} & \textbf{p-value} \\ \hline
AR(3) Model                  & 0.85 & (0.7278, 0.9858) & 0.0127 \\
ST Model-1                   & 0.78 & (0.6163, 0.9126) & 0.0182 \\
ST Model-2 ($f(d)=1/d$)      & 0.09 & (0,      0.0972) & 0.3578 \\
ST Model-2 ($f(d)=\exp(-d)$) & 0.14 & (0,      0.1530) & 0.3103 \\ \hline
\end{tabular}
\end{center}
\end{table}

For the spatio-temporal-model-2  (with either distance function), we are unable to reject the null hypothesis of no spatial association.
Thus we can conclude that the residuals obtained from this model are spatially pairwise independent. Therefore Model (\ref{model}) satisfactorily explains the spatial and temporal incidence pattern of COVID-19 in the state Kerala.

Based on the overall conclusion on spatial association (lack thereof in the residuals after applying the spatio-temporal model with the exponential decay model \ref{model}), we further investigated the presence of individual pairwise dependence, if any. Accordingly in Figure \ref{fig:corplotst2} we present the pairwise Bergsma correlations (see \cite{bergsma2006new}) for the 14 districts based the residuals obtained from Model \ref{model}.
\begin{figure}[ht]
\includegraphics[width=0.7\linewidth]{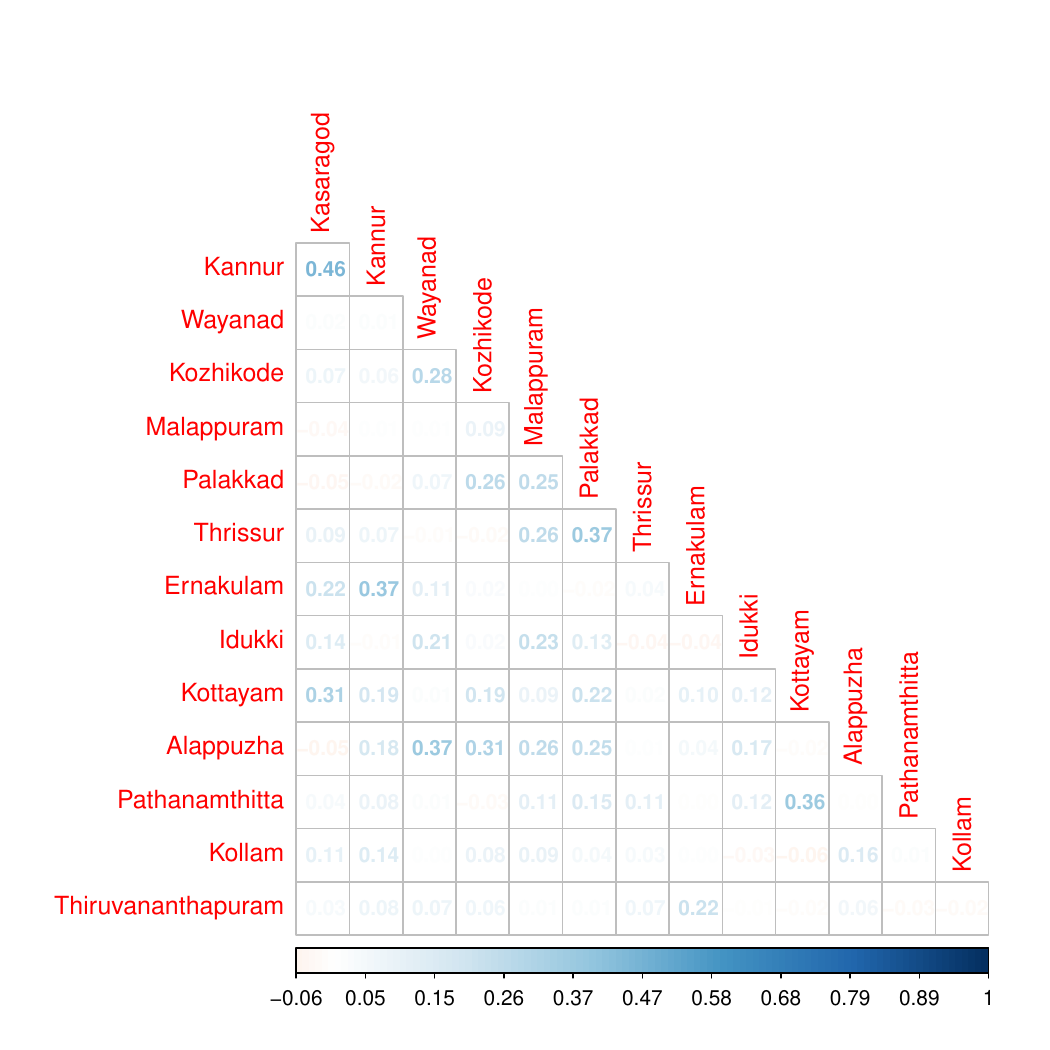}
\caption{Pairwise Bergsma correlations of the residuals from the spatio-temporal Model \eqref{model}} with exponential distance decay.
\label{fig:corplotst2}
\end{figure}
As was shown for the $\tilde{S}_B$ statistic, for the $\tilde{\rho}$ statistics also it is possible to assess significance (using an empirically derived cutoff of $\tilde{\rho}>0.17$). By applying such a test, we conclude pairwise independence of most pairs of districts. There appears to be a few sporadic significant pairs, which may be due to some yet unknown factor(s).

\section{Discussion}
We have used Bergsma's correlation and its estimate to define a global spatial measure of association. This measure is asymptotically normal when we have observations that are i.i.d.~over time but are spatial dependent. In the absence of spatial dependence, more precisely when there is only pairwise spatial independence, the estimate has an asymptotic distribution involving infinite sum of weighted i.i.d.~chi-square variables. In both cases, the distribution depends on the unknown underlying distribution of the observations, especially through the eigenvalues of appropriate kernel functions. Simulations show that for reasonably large number of observations, the actual finite sample distributions are well approximated by the asymptotic distributions. It is also seen that these measures and their distributions are sensitive to the choice of a spatial proximity matrix. 

At present no distributional properties are known for the association estimate when the observations are also temporally dependent. Inspite of this, the measure can be used for model fitting purposes even when such spatial dependence is present. This is done by first removing the temporal dependence through appropriate modelling and then using the measure on the residual series. We have presented such an application for spatio-temporal modelling of COVID-19 data on the monthly time series data for the 14 districts of the Indian state of Kerala.  

In the simulation work here, we have explored a somewhat smallish number of regions. As otherwise small signal (of dependence) may be lost in a large collection of otherwise independent locations. We have also simulated from a moderately large number of spatial locations $R=99$. The null behaviour remains the same, with predominant factor affecting the distribution being the proximity matrix $W$ (see Figures \ref{fig:nullsims5st} and \ref{fig:nullcomprsns5st}). The behaviour under the alternate scenario remains same under the lag-1 adjacency $W$ matrix under both SMA and SAR models (see Figure \ref{fig:boxplots5st}). However there's clear loss of sensitivity under the inverse distance $W$ matrix (see Appendix B for illustations).

For other potential extensions of the current work one may refer to \cite{jkss2023}.

\bibliographystyle{abbrvnat}

\section*{Appendix A: {\tt R} code}
\appendix
\renewcommand{\thesection}{\Alph{section}}


We give below two {\tt R} codes.
The first is for the $U$-statistic based estimate of Bergsma correlatio.
Input for this first code is bivariate data in the form of two vectors, $x$ and $y$.
The second is a code for computing the $\tilde{S}_B$ statistic. For this data is to be given in the format of a matrix with $T$ rows (time points) and $R$ columns (locations), and an $R \times R$ Spatial proximigty matrix.\vskip10pt

\noindent Bergsma's ${\tilde \rho}$

\begin{verbatim}
rho_tilde <- function(x,y)
{
n   = length(x)
X   = matrix( replicate(n, x), nrow = n); Dx = abs( X - t(X))
Y   = matrix( replicate(n, y), nrow = n); Dy = abs( Y - t(Y))
A1  = apply(Dx,1,mean)
A2  = apply(Dy,1,mean)
B1  = mean(A1)
B2  = mean(A2)
Dx1 = sweep(Dx,1,(n/(n-1))*A1)
Dx2 = sweep(Dx1,2,(n/(n-1))*A1)
Hx  = (-1/2)*(Dx2 +(n/(n-1))*B1)
Dy1 = sweep(Dy,1,(n/(n-1))*A2)
Dy2 = sweep(Dy1,2,(n/(n-1))*A2)
Hy  = (-1/2)*(Dy2 +(n/(n-1))*B2)
Hxy = Hx*Hy
I   = matrix(1,n,n)
I[lower.tri(I,diag=TRUE)] = 0
r_tilde = sum(I*Hxy)/(sqrt(sum(I*Hx*Hx)*sum(I*Hy*Hy))) #computing rho_tilde
return(r_tilde)
}
\end{verbatim}

\vspace{1cm}
\noindent $\tilde{S}_B$ statistic

\begin{verbatim}
SB <- function(data, W){
r_curl = matrix(0, ncol(data), ncol(data))
for(i in 1:ncol(data)){
for(j in i:ncol(data)){
r_curl[i,j] = rho_tilde(data[,i], data[,j])
r_curl[j,i] = r_curl[i,j]
}}
Spatial_AI = sum(W*r_curl)/sum(W)
return(Spatial_AI)
}
\end{verbatim}

\section*{Appendix B: Simulation with $R=99$ locations}

\begin{figure}[ht]
\centering
\includegraphics[width=1\linewidth]{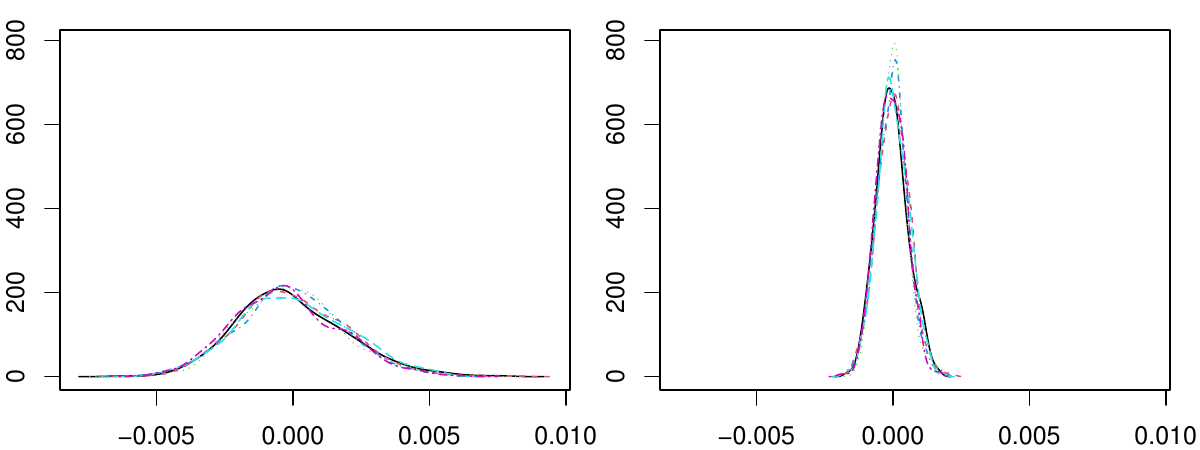}
\caption{Simulated null distribution, $T=50$, $R=99$ for $500$ replicates, for six distributions. Proximity matrices $W$: lag-$1$ adjacency (left), 
inverse distance (right). }
\label{fig:nullsims5st}
\end{figure}

\begin{figure}[ht]
\centering
\includegraphics[width=1\linewidth]{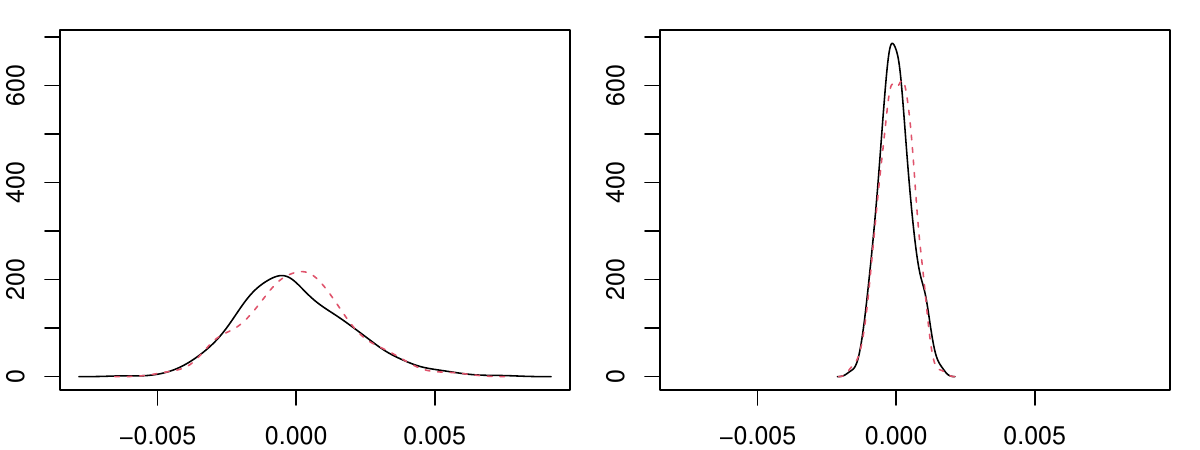}
\caption{Comparison of simulated null (black-solid line) and asymptotic null distributions (red-dashed-line) of $T*S_B$ with $500$ replicates, based on finite discrete approximation using 100 eigenvalues for aymptotic distribution, with $W$ matrices: lag-$1$ adjacency (left), inverse distance (right). }
\label{fig:nullcomprsns5st}
\end{figure}

\begin{figure}[ht]
\centering
\includegraphics[width=0.8\linewidth]{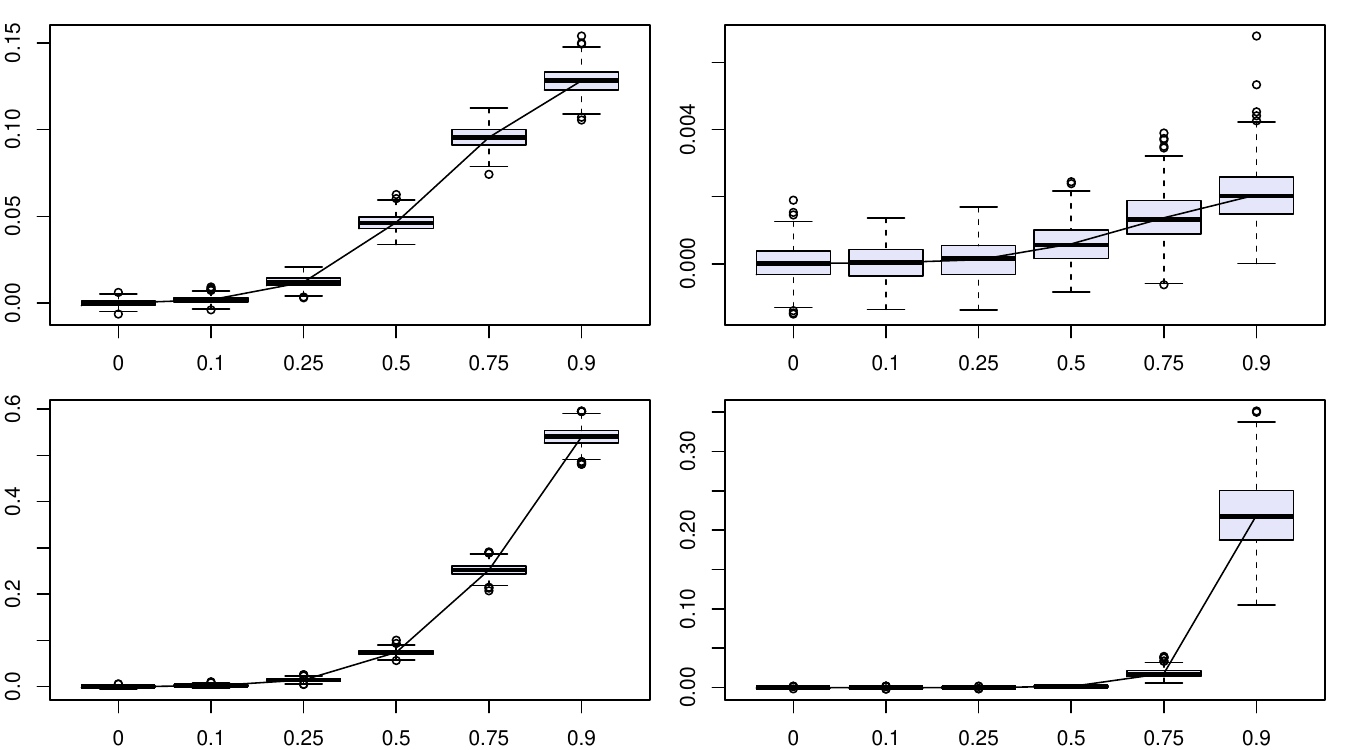}
\caption{Empirical distributions of $S_{B}$, $R$=99, $T$=50, $500$ replicates, with sample mean overlayed, dependence parameter $\theta$ on $x$-axis under SMA model (top), SAR model (bottom). Proximity matrices: lag-$1$ adjacency (left), and inverse distance (right).}
\label{fig:boxplots5st}
\end{figure}

\end{document}